\documentclass[11pt,a4paper]{article}
\usepackage{cmap}
\usepackage[T1]{fontenc}
\usepackage[utf8]{inputenc}
\usepackage{lmodern}
\ifdefined\pdfgentounicode
  \input{glyphtounicode}
\fi
\usepackage{geometry}
\usepackage{setspace}
\usepackage{microtype}
\usepackage{amsmath,amssymb,amsthm}
\usepackage{booktabs,longtable,tabularx,array}
\usepackage{enumitem}
\usepackage{hyperref}
\usepackage{caption}
\hypersetup{colorlinks=true,linkcolor=black,citecolor=black,urlcolor=black,
  pdftitle={When AI Generates Covariates: Causal Typing and Estimand Drift in Sequential Experiments},
  pdfauthor={Takes Fujita and Nobutaka Hattori},
  pdfsubject={Causal typing and estimand drift for generated covariates in sequential experiments},
  pdfkeywords={causal inference; artificial intelligence; generated covariates; estimand drift; causal typing; sequential experiments}}
\setlist[itemize]{leftmargin=2em}
\setlist[enumerate]{leftmargin=2em}
\newtheorem{definition}{Definition}
\newtheorem{assumption}{Assumption}
\newtheorem{theorem}{Theorem}
\newtheorem{proposition}{Proposition}
\newtheorem{corollary}{Corollary}
\newcommand{\indep}{\mathrel{\perp\!\!\!\perp}}
\newcommand{\E}{\mathbb{E}}

\newcommand{\calG}{\mathfrak{G}}
\newcommand{\calR}{\mathcal{R}}
\newcommand{\calC}{\mathcal{C}}
\newcommand{\calS}{\mathcal{S}}
\newcommand{\calL}{\mathcal{L}}
\newcommand{\id}{\mathrm{id}}
\newcommand{\obs}{\mathrm{obs}}
\newcommand{\emp}{\mathrm{emp}}
\newcommand{\supop}{\mathrm{sup}}
\newcommand{\true}{\mathrm{true}}
\newcommand{\expit}{\mathrm{expit}}
\newcolumntype{Y}{>{\raggedright\arraybackslash}X}

\title{When AI Generates Covariates:\\Causal Typing and Estimand Drift in Sequential Experiments}
\author{Takes Fujita\textsuperscript{1} \quad Nobutaka Hattori\textsuperscript{2}\\[4pt]
\small \textsuperscript{1}VRI\\
\small \textsuperscript{2}Department of Neurology, Juntendo University School of Medicine}
\date{}

\begin{document}
\maketitle

\begin{abstract}
AI-generated covariates from notes, conversations, images, and wearable streams can change the causal question when their roles are left unspecified. A generated feature may represent a treatment version, pre-action state, history, design variable, mediator, outcome proxy, observation process, or intercurrent event; these roles are not interchangeable.

We formulate a causal type discipline for sequential experiments: a versioned representation map, a causal role classifier, a claim-status filter, and an estimand lock. The lock fixes a standardized proximal effect before generated covariates enter the analysis. Under audit correctness and standard identification assumptions, admissible role assignments preserve this estimand. We apply the established conditional-covariance characterization of compression bias to substitution of generated representations for design-relevant states. A standardized decomposition separates compression, conditional-law, and standardization drift. Further results cover mediator adjustment, post-action leakage, marker-intervention conflation, outcome-guided discovery, and state-measurement error. Cluster-level orthogonal estimators distinguish empirical and superpopulation targets under repeated sessions and missing outcomes.

Simulations show that refinement helps when it retains design-relevant information, whereas design erasure, leakage, and same-data marker selection can produce bias or undercoverage. The framework places causal semantics and claim status before confirmatory inference with generated representations.
\end{abstract}


\section{Introduction}

Modern causal analyses increasingly use variables that did not exist as simple fields in a case report form, registry, database, or randomized-trial data set. Investigators generate covariates from clinical notes, diaries, speech, images, food logs, wearable streams, environmental records, mobile-phone traces, and human-machine interaction logs. These variables may be produced by a human coding protocol, a rule-based dictionary, a supervised model, a foundation model, a wearable-signal algorithm, a causal representation learner, or a hybrid human-AI workflow. Although AI-generated covariates are the motivating case, the framework applies to this broader class of generated representation maps. Once generated, they are often treated as ordinary covariates.

That step is not innocuous. A generated covariate has a causal role. A variable extracted before a decision point may be a treatment version, core state, candidate state marker, prior history, or design variable. A variable extracted after the decision point may be a mediator, outcome proxy, intercurrent event, observation indicator, or post-outcome descriptor. A variable selected because it was associated with the outcome may be useful for exploration but cannot support a confirmatory effect-modification claim unless selection is accounted for. Prediction does not decide causal admissibility.

The motivating example is a repeated-session clinical experiment in which a current-session explanation, cue, device context, or behavioral support is randomized and a proximal outcome is measured shortly afterward. The same participant contributes several sessions. The current-session effect may depend on medication phase, time since dose, fatigue, recent sleep, prior response, expectation, environmental context, or recent activity. AI systems can decompose diaries, clinical notes, conversations, and wearable streams into candidate state variables. The central question, however, is not how to extract more variables. It is how to prevent extracted covariates from changing the estimand without notice.

Although the example is clinical, the object of this paper is theoretical. We study generated covariates in sequential causal experiments. The paper develops a local causal-role theory: it asks whether a generated variable may play its proposed role under a locked estimand. This question precedes generated-regressor inference, which addresses downstream uncertainty after a role has already been accepted.

Typing is not certification. Typing states the role a generated variable would need to play in order to be admissible for a locked claim. Certification concerns the evidence that the role is credible in a particular data environment. This paper studies generated covariates one role at a time; the broader problem of certifying many generated variables as a reconstructed pre-action information set is left to separate work.

The main contributions are as follows.
\begin{enumerate}
\item It defines generated covariates as outputs of versioned, possibly stochastic representation maps from longitudinal event streams and assigns causal types according to temporal measurability, intervention status, design role, and downstream role.
\item It defines an estimand lock: a target causal contrast, target state distribution, outcome time, intercurrent-event strategy, missingness strategy, and claim status fixed before generated covariates are used for confirmatory inference.
\item It states admissibility rules for treatment contrasts, state-design objects, observation strategies, intercurrent-event handling, claim status, and versioning, giving causal typing operational consequences.
\item It proves an estimand-preservation theorem: under standard causal identification assumptions, analyses satisfying the admissibility rules identify the locked standardized proximal effect.
\item It applies established deconfounding-score identities \cite{damour2021,clivio2026} to representation substitution in sequential experiments. The conditional form records which design information a generated covariate may discard; a standardized decomposition separates compression, conditional-law, and standardization drift. The contribution is the integration of these results with causal types, an estimand lock, and audit rules.
\item It derives an orthogonal influence-function estimator for the locked effect under clustered repeated-session data and missing outcomes, and gives a simulation in which performance is organized around type violations rather than around disease-specific efficacy.
\end{enumerate}

The framework complements potential-outcome causal inference, longitudinal g-methods, structural nested models, micro-randomized trials, dynamic treatment regimes, mediation analysis, double/debiased machine learning, targeted learning, post-selection inference, and causal representation learning \cite{hernan2020,pearl2009,robins1986,robins2000,vansteelandt2014,klasnja2015,boruvka2018,qian2022,murphy2003,vanderweele2015,chernozhukov2018,vdl2006,kuchibhotla2022}. Its distinct question comes before estimation: which generated covariates are admissible for the causal claim being made? A brief secondary use is developed later: type-violation reasons can be retained as exploratory follow-up signals, without converting the framework into a causal discovery algorithm.

\section{Sequential experiment and generated representations}

Let \(i=1,\ldots,N\) index independent participants or clusters. Participant \(i\) contributes sessions \(s=1,\ldots,S_i\). We suppress the proximal outcome time \(r\) when it is not needed; multiple outcome times are handled by fixed weights later.

Each session has a pre-action information set \(\mathcal F^{-}_{is}\), an action time \(t^A_{is}\), a current action \(A_{is}\in\{0,1\}\), a post-action information set, an observation indicator \(R_{is}\), and an outcome \(Y_{is}\) when \(R_{is}=1\). The pre-action information set may contain baseline variables, prior outcomes, prior treatments, prior intercurrent events, design variables, and raw event streams recorded before the current action. The action \(A_{is}\) is a treatment version, not merely a label. For example, \(A=1\) might denote a standardized high-expectation explanation and \(A=0\) a neutral explanation, with contact duration and measurement schedule held fixed.

Let \(Y_{is}(a)\) denote the potential proximal outcome under current-session treatment version \(a\), with the past history up to \(\mathcal F^{-}_{is}\) left as observed. Carryover from earlier sessions is not assumed absent; it must be represented in the pre-action history or handled through the intercurrent-event strategy.

\begin{definition}[Generated representation]
A generated representation at session \((i,s)\) is a vector
\[
Z^{(m)}_{is}=\calR_m(E_{is},\xi^{(m)}_{is}),
\]
where \(E_{is}\) is a longitudinal event stream, \(\calR_m\) is a versioned representation map, and \(\xi^{(m)}_{is}\) denotes possible algorithmic randomness. The map \(\calR_m\) may be a manual coding rule, a conventional feature extractor, an AI system, or a human-AI pipeline. Its version includes the model, extraction instructions or dictionary, input window, feature schema, preprocessing rules, quality-control rules, and date or hash.
\end{definition}

A representation is pre-action if it is measurable with respect to \(\mathcal F^{-}_{is}\), up to external algorithmic randomness independent of current potential outcomes conditional on the raw pre-action stream. A representation is post-action if its input window or construction uses information after \(A_{is}\) is assigned or delivered. A representation is mixed-window if it combines pre-action and post-action inputs without a separable audit trail. For primary total-effect analysis, a mixed-window representation is treated as post-action unless the pre-action component can be reconstructed and versioned separately.

In practice there may be more than one plausible action milestone: recommendation formation, assignment recording, patient communication, and intervention delivery can occur at different times. The estimand lock must name the milestone relevant to the causal contrast. A covariate whose source window overlaps competing milestones is treated as mixed-window for the primary lock unless a separable pre-milestone component can be reconstructed. When the milestone itself is scientifically ambiguous, analysts should report separate locks or a source-window sensitivity analysis rather than silently choosing the most convenient boundary.

All identification and drift statements below are conditional on a fixed recorded version of the representation map. If \(\calR_m\) uses algorithmic randomness, the analysis must either record the seed, condition on the generated output as produced by the locked version, or declare an external seed distribution before outcome analysis. Equivalently, the random seed can be included in the pre-action object whenever it is independent of current potential outcomes conditional on the raw pre-action stream. If neither the seed nor the stochastic generation rule is reproducible, the versioning rule is violated.

Temporal order is necessary but not sufficient. A pre-action variable may be a randomized treatment version, design variable, core state, candidate state marker, or prior history. A post-action variable may be a mediator, outcome proxy, intercurrent event, or measurement descriptor. The type discipline below records these distinctions.

\section{Causal types, claim status, and estimand locks}

\begin{definition}[Causal type discipline]
A causal type discipline for generated covariates is
\[
\calG=(\calR_m,\calC,\calS,\calL),
\]
where \(\calR_m\) is a versioned representation map, \(\calC\) is a causal type classifier, \(\calS\) is a claim-status filter, and \(\calL\) is an estimand lock.
\end{definition}

For a generated covariate \(Z_{isj}\), the classifier assigns
\[
\calC(Z_{isj})\ \text{belongs to}\ \{A,B,U,H,D,M,Y,\mathrm{Obs},I,\mathrm{Excluded}\}.
\]
The roles are defined in Table \ref{tab:types}.

\begin{longtable}{>{\raggedright\arraybackslash}p{0.13\linewidth}>{\raggedright\arraybackslash}p{0.78\linewidth}}
\caption{Causal types for generated covariates.}\label{tab:types}\\
\toprule
Type & Meaning \\
\midrule
\endfirsthead
\toprule
Type & Meaning \\
\midrule
\endhead
\(A\) & Current-session treatment version assigned or otherwise defined as an intervention. \\
\(B\) & Core pre-action state required by the scientific question or design. \\
\(U\) & Candidate pre-action state marker or effect modifier generated from event streams. \\
\(H\) & Prior history, including previous treatments, responses, adverse events, and timing. \\
\(D\) & Design information, such as strata, allocation probability, eligibility, availability, block, site, or adaptive randomization state. \\
\(M\) & Post-action mediator candidate or mechanism variable. \\
\(Y\) & Outcome, outcome proxy, or post-outcome descriptor. \\
Obs & Observation, missingness, censoring, or measurement indicator. \\
\(I\) & Intercurrent event requiring a strategy for the estimand. \\
Excluded & Excluded covariate, non-causal descriptor, leakage variable, or covariate outside the analysis scope. \\
\bottomrule
\end{longtable}

The claim-status filter assigns
\[
\calS(Z_{isj})\in\{\mathrm{confirmatory},\mathrm{secondary},\mathrm{exploratory},\mathrm{next\ trial},\mathrm{not\ causal}\}.
\]
A covariate of type \(U\) can support pre-specified or secondary effect-modification analysis. It does not by itself support a claim that intervening on \(U\) would change \(Y\). A covariate of type \(M\) can support mechanism analysis, but it is not an admissible adjustment variable for the primary total effect. A covariate discovered or selected using outcome information is exploratory unless the selection step is separated from or incorporated into inference.

\begin{definition}[Estimand lock]
An estimand lock is a tuple
\[
\calL=(a,a',r,V,Q^*,\omega,\mathcal I,\mathcal R_{obs},\kappa),
\]
where \(a\) and \(a'\) are treatment versions, \(r\) is the proximal outcome time or a set of times, \(V\) is the admissible pre-action state-design object, \(Q^*\) is the target distribution of \(V\), \(\omega\) gives outcome-time weights when needed, \(\mathcal I\) specifies the intercurrent-event strategy, \(\mathcal R_{obs}\) specifies the missingness or observation strategy, and \(\kappa\) is the permitted claim status.
\end{definition}

In the simplest single-time case, with \(V_{is}=(X_{is},D_{is})\) containing pre-action state and design information, the locked standardized proximal effect is
\begin{equation}
\Psi(a,a';Q^*)=\int\{m_a(v)-m_{a'}(v)\}\,dQ^*(v),
\label{eq:psi}
\end{equation}
where
\[
m_a(v)=\E\{Y_{is}(a)\mid V_{is}=v\}.
\]
For multiple proximal outcome times \(r\in\mathcal R\), define
\[
\Psi_\omega(a,a';Q^*)=\sum_{r\in\mathcal R}\omega_r\int\{m_{ar}(v)-m_{a'r}(v)\}\,dQ^*(v),\qquad \sum_r\omega_r=1.
\]
The target distribution \(Q^*\) must be specified before outcome modeling. It may be participant-weighted, session-weighted, restricted to an eligibility region, or transported to an external state distribution. If session counts vary across participants, participant-weighted and session-weighted estimands are different. This distinction is part of the lock, not a detail to be decided after results are seen. If \(Q^*\) is external to the study, support and transportability conditions must be stated; otherwise the lock defines an estimand that the observed study cannot identify without additional assumptions.

\section{Admissibility rules for generated covariates}

An analysis must specify the treatment contrast, the state-design object used for adjustment or standardization, the target distribution, the observation strategy, the intercurrent-event strategy, the nuisance functions, and the status of the claim being made. We call such an analysis admissible for an estimand lock when the variables placed in these positions satisfy the role-compatibility rules in Table \ref{tab:rules}. These rules state what follows from a proposed role assignment; they do not certify that the assignment is correct or that the classifier \(\calC\) can be learned automatically from the data. When evidence for a role is weak, the covariate should be excluded from the confirmatory state-design object, assigned a conservative downstream role, or analyzed under an alternative estimand lock.

\begin{longtable}{>{\raggedright\arraybackslash}p{0.18\linewidth}>{\raggedright\arraybackslash}p{0.35\linewidth}>{\raggedright\arraybackslash}p{0.37\linewidth}}
\caption{Core admissibility rules for generated covariates.}\label{tab:rules}\\
\toprule
Rule & Required type condition & Consequence for the locked analysis \\
\midrule
\endfirsthead
\toprule
Rule & Required type condition & Consequence for the locked analysis \\
\midrule
\endhead
Treatment & Treatment positions contain only variables of type \(A\) with operational treatment versions. & The contrast is about changing the declared treatment version, not an observed descriptor. \\
Pre-action state & The primary state-design object \(V\) contains only pre-action variables of type \(B,U,H,D\). & Current-session descendants of \(A\) cannot define the state for a primary total effect. \\
Design retention & Variables of type \(D\) that affect assignment, eligibility, availability, observation, or target construction are retained in \(V\) or represented in a known design mechanism. & Randomization, positivity, and observation arguments remain attached to the variables that generated them. \\
Total effect & Variables of type \(M\) are excluded from the primary total-effect adjustment set. & Mechanism variables require a separate mediation lock. \\
Outcome and observation & Variables of type \(Y\), post-action Obs, and post-action \(I\) do not define pre-action state. & They enter only through the outcome, observation model, or intercurrent-event strategy. \\
Claim status & Outcome-guided covariates cannot support confirmatory effect-modification claims without pre-specification, sample splitting, selective inference, or simultaneous inference. & Exploratory discovery is not reported as confirmatory heterogeneity. \\
Versioning & A covariate may enter confirmatory analysis only through the version of \(\calR_m\) recorded in the lock or audit object. & Post-hoc extraction-instruction changes, dictionary edits, model updates, or unreproducible stochastic extraction create a new representation version. \\
\bottomrule
\end{longtable}

These rules are intentionally restrictive. Their purpose is not to maximize prediction. Their purpose is to preserve the causal quantity in Equation \eqref{eq:psi}. Table \ref{tab:violations} summarizes common violations and the resulting drift.

\begin{longtable}{>{\raggedright\arraybackslash}p{0.24\linewidth}>{\raggedright\arraybackslash}p{0.34\linewidth}>{\raggedright\arraybackslash}p{0.32\linewidth}}
\caption{Common type violations and the causal quantity they tend to change.}\label{tab:violations}\\
\toprule
Violation & What is wrong & Typical consequence \\
\midrule
\endfirsthead
\toprule
Violation & What is wrong & Typical consequence \\
\midrule
\endhead
Mediator used as state & A post-action mechanism variable is adjusted for in a primary total-effect model. & The target moves from a total effect to a direct, path-blocked, or otherwise conditioned contrast. \\
State marker reported as intervention & Heterogeneity across observed states is interpreted as the effect of manipulating the state. & A next-trial hypothesis is reported as a confirmatory intervention claim. \\
Design erasure & A predictive representation replaces design information used for assignment, eligibility, availability, or observation. & Exchangeability or positivity may fail after compression. \\
Post-action leakage & Generated covariates from post-treatment text, sensors, or notes are inserted into the pre-action state. & The comparison conditions on treatment-specific strata. \\
Outcome-guided discovery & The same outcomes are used to select and test an effect modifier. & Nominal intervals and tests no longer support confirmatory claims. \\
\bottomrule
\end{longtable}

\section{Identification and estimand preservation}

The assumptions below separate statistical identification from audit correctness. The type discipline does not make the classifier \(\calC\) correct by definition; it makes the required judgement explicit and records the estimand consequences of that judgement.

\begin{assumption}[Audit correctness and version stability]
For the locked claim, the audit correctly records the source window, representation version, causal type, and claim status of each generated covariate admitted to the analysis. Variables admitted to the primary state-design object are pre-action for the current treatment version, design variables required for assignment, eligibility, availability, observation, or target construction are not erased without a declared compression argument, and the representation version is reproducible under the lock.
\end{assumption}

\begin{assumption}[Consistency and treatment versions]
If \(A_{is}=a\), then \(Y_{is}=Y_{is}(a)\) whenever the outcome is observed. The treatment versions \(a\) and \(a'\) are operationally defined.
\end{assumption}

\begin{assumption}[Sequential exchangeability]
For each session in the target region, \(A_{is}\) is independent of \(\{Y_{is}(1),Y_{is}(0)\}\) given \(V_{is}\). This may be supported by randomization conditional on \(V_{is}\), by design-based assignment probabilities, or by substantive assumptions in observational sequential studies.
\end{assumption}

\begin{assumption}[Positivity]
For all \(v\) in the support of \(Q^*\), \(0<P(A_{is}=1\mid V_{is}=v)<1\).
\end{assumption}

\begin{assumption}[History and interference]
Between-participant interference is absent or negligible. Within-participant carryover is represented in \(V_{is}\) through prior history \(H_{is}\) or through a declared intercurrent-event strategy.
\end{assumption}

\begin{assumption}[Observation process]
When outcomes may be missing, the observation indicator satisfies \(R_{is}\) independent of \(\{Y_{is}(1),Y_{is}(0)\}\) given \(A_{is}\) and \(V_{is}\), with \(P(R_{is}=1\mid A_{is}=a,V_{is}=v)>0\) in the target region. If this condition is not credible, sensitivity parameters rather than complete-case claims are required.
\end{assumption}

\noindent\textbf{Observation-sensitivity note.} Assumption 6 is an identifying condition, not a consequence of causal typing. If outcome observation is not conditionally ignorable, define \(\rho_a(v)=P(R=1\mid A=a,V=v)\) and let \(\delta_a(v)\) be the mean difference between unobserved and observed potential outcomes in arm \(a\) at \(v\). Then, under consistency and treatment exchangeability given \(V\),
\[
m_a(v)=\mu_a(v)+\{1-\rho_a(v)\}\,\delta_a(v).
\]
Thus failure of observation exchangeability moves the analysis from a point claim to a sensitivity family rather than changing the causal type discipline.

\begin{theorem}[Estimand preservation under admissible causal-role assignments]
Consider a locked estimand \(\calL\) with target effect \(\Psi(a,a';Q^*)\) in Equation \eqref{eq:psi}. Suppose the variables used for the treatment contrast, state-design object, observation strategy, intercurrent-event strategy, nuisance functions, and claim status satisfy the admissibility rules in Table \ref{tab:rules}, and Assumptions 1-6 hold. Then the observed-data functional
\begin{equation}
\Psi^{\id}(a,a';Q^*)=\int\{\mu_a(v)-\mu_{a'}(v)\}\,dQ^*(v),
\label{eq:identifying}
\end{equation}
where
\[
\mu_a(v)=\E(Y_{is}\mid A_{is}=a,R_{is}=1,V_{is}=v),
\]
identifies the locked standardized proximal effect. In the absence of missing outcomes, the same result holds with \(R\) omitted.
\end{theorem}

\begin{proof}
Because the admissibility rules hold, \(V_{is}\) is pre-action and contains the state and design information required by the lock; no current-session mediator, outcome proxy, or post-action descriptor is included as if it were pre-action state. By consistency, among observed outcomes with \(A_{is}=a\), \(Y_{is}=Y_{is}(a)\). By the observation-process assumption, conditioning on \(R_{is}=1\) does not change the conditional mean of \(Y_{is}(a)\) given \(A_{is}=a,V_{is}=v\). By sequential exchangeability,
\[
\E\{Y_{is}(a)\mid A_{is}=a,V_{is}=v\}=\E\{Y_{is}(a)\mid V_{is}=v\}.
\]
Positivity ensures the conditional means are defined on the support of \(Q^*\). Therefore \(\mu_a(v)=m_a(v)\) and \(\mu_{a'}(v)=m_{a'}(v)\) for \(Q^*\)-almost every \(v\). Substituting these equalities into Equation \eqref{eq:identifying} gives Equation \eqref{eq:psi}.
\end{proof}

Once generated covariates are assigned roles compatible with the lock, standard causal identification machinery applies. The added requirement is the causal-role discipline that prevents generated representations from silently changing the target.

\section{Estimand drift under type violations}

Type preservation is useful only if violations have consequences. This section records generic ways in which generated covariates can change the estimand or the inferential claim. The central violation is substituting a predictive representation for design-relevant causal information.

\subsection{No-free substitution and compression drift}

Let \(W\) be an admissible pre-action state-design object satisfying exchangeability. Let \(Z=T(W)\) be a generated representation that an analyst uses instead of \(W\). If the representation map is stochastic, take \(W\) in this subsection to mean the augmented object \((W,\xi)\) after the seed or stochastic generation rule has been fixed by the lock. Define
\[
\pi_a(W)=P(A=a\mid W),\qquad m_a(W)=\E\{Y(a)\mid W\}.
\]

The conditional-covariance mechanism below is established in the deconfounding-score literature. D'Amour and Franks \cite{damour2021}, Propositions 1--2 and Appendix B, characterize reduction bias using outcome and assignment covariances; Clivio et al. \cite{clivio2026}, Lemma 3.1, give a density-ratio formulation for the average treatment effect on the treated. We state the arm-specific conditional identity in the present notation to connect that mechanism to sequential design information and the estimand lock. The covariance identity and its zero-covariance criterion are not claimed as new general compression theory.

\begin{theorem}[No-free-substitution: conditional compression drift]
Assume consistency and \(A\indep Y(a)\mid W\). For any representation \(Z=T(W)\) with \(\E\{\pi_a(W)\mid Z=z\}>0\),
\begin{equation}
\E(Y\mid A=a,Z=z)-\E\{Y(a)\mid Z=z\}
=
\frac{\operatorname{Cov}\{m_a(W),\pi_a(W)\mid Z=z\}}{\E\{\pi_a(W)\mid Z=z\}}.
\label{eq:drift}
\end{equation}
Therefore conditioning on \(Z\) instead of \(W\) generally changes the conditional causal mean. For a treatment contrast \(a\) versus \(a'\), the conditional drift is the difference of the two arm-specific covariance terms.
\end{theorem}

\begin{proof}
By consistency and exchangeability given \(W\),
\[
\E(Y\mid A=a,Z=z)=\E\{m_a(W)\mid A=a,Z=z\}.
\]
Bayes weighting within levels of \(Z\) gives
\[
\E\{m_a(W)\mid A=a,Z=z\}=
\frac{\E\{m_a(W)\pi_a(W)\mid Z=z\}}{\E\{\pi_a(W)\mid Z=z\}}.
\]
Subtracting \(\E\{m_a(W)\mid Z=z\}\) yields Equation \eqref{eq:drift}.
\end{proof}

The identity shows that a generated representation \(Z\) may predict \(Y\) well and still fail to preserve the causal target if omitted information affects both assignment and potential-outcome means within levels of \(Z\).

\begin{corollary}[Safe compression]
Under the conditions of Theorem 2, substituting \(Z=T(W)\) for \(W\) preserves the conditional mean for arm \(a\) at \(Z=z\) if and only if
\[
\operatorname{Cov}\{m_a(W),\pi_a(W)\mid Z=z\}=0.
\]
Two interpretable sufficient conditions are: assignment compression, \(\pi_a(W)=\pi_a(Z)\) almost surely within \(Z=z\); and outcome compression, \(m_a(W)=m_a(Z)\) almost surely within \(Z=z\). For a treatment contrast, preservation holds exactly when the difference of the two arm-specific drift terms is zero. Zero drift in each arm is sufficient; cancellation can also preserve that particular contrast.
\end{corollary}

\begin{corollary}[Predictive sufficiency is not causal sufficiency]
Suppose \(Z=T(W)\) and that \(\operatorname{Var}\{m_a(W)\mid Z\}>0\) with positive probability. Holding fixed the joint distribution of \((W,Y(a))\) and the representation map \(T\), one can construct two treatment-assignment mechanisms with the same potential-outcome prediction problem from \(Z\): one with zero compression drift everywhere and one with strictly positive compression drift on a set of representation values of positive probability. Thus no criterion based only on prediction risk from \(Z\) can certify that \(Z\) is a valid substitute for \(W\) in the causal analysis.
\end{corollary}

\begin{proof}
The first mechanism sets \(\pi^{(0)}_a(W)=\tfrac12\), so the covariance in Equation \eqref{eq:drift} is zero for every \(z\). For the second, fix \(\epsilon\in(0,\tfrac12)\) and set
\[
\pi^{(1)}_a(W)=\tfrac12+\epsilon\tanh[m_a(W)-\E\{m_a(W)\mid Z\}].
\]
Because \(|\tanh|<1\), the assignment probability lies in \((0,1)\) without requiring \(m_a(W)\) to be bounded. For every \(z\), the conditional covariance in Equation \eqref{eq:drift} is \(\epsilon\operatorname{Cov}(m_a(W),\tanh[m_a(W)-\E\{m_a(W)\mid Z\}]\mid Z=z)\), which is strictly positive whenever \(\operatorname{Var}\{m_a(W)\mid Z=z\}>0\), because the covariance between a nondegenerate random variable and a strictly increasing function of itself is strictly positive. By assumption, this event has positive probability under the distribution of \(Z\), so the second mechanism produces nonzero compression drift on a set of representation values of positive probability. The potential-outcome prediction problem from \(Z\) has not changed because the distribution of \((W,Y(a))\) and the map \(T\) have not changed.
\end{proof}

\begin{corollary}[Standardized substitution drift]
Under the conditions of Theorem 2 under \(P\), two probability laws must be distinguished. Let \(P\) denote the observed-data law of \((W,A,Z,Y)\), under which the conditional means \(\widetilde\mu_a(z)=\E_P(Y\mid A=a,Z=z)\) and \(\bar m^P_a(z)=\E_P\{m_a(W)\mid Z=z\}\) are defined. Separately, let \(Q^*\) denote the target law obtained by equipping \(W\) with the locked target distribution \(Q_W^*\) and generating \(Z\) from the locked representation map, define \(\bar m^*_a(z)=\E_{Q^*}\{m_a(W)\mid Z=z\}\), and let \(Q^*_Z\) denote the push-forward of \(Q_W^*\) under that map. Let \(Q_Z^\dagger\) denote the standardizing distribution over \(Z\) that the analyst actually uses; it need not equal \(Q^*_Z\). Assume the support condition that \(Q_Z^\dagger\) is dominated both by the observed-law distribution of \(Z\) restricted to \(\{z:\E_P\{\pi_a(W)\mid Z=z\}>0\}\) and by \(Q^*_Z\), so that \(\widetilde\mu_a\), \(\bar m^P_a\), and \(\bar m^*_a\) are all defined \(Q_Z^\dagger\)-almost everywhere. If an analyst reports
\[
\widetilde\psi^Z_a=\int \widetilde\mu_a(z)dQ_Z^\dagger(z)
\]
in place of the locked arm-specific target
\[
\psi^W_a=\int m_a(w)dQ_W^*(w),
\]
then
\begin{align}
\widetilde\psi^Z_a-\psi^W_a
&=\int\{\widetilde\mu_a(z)-\bar m^P_a(z)\}\,dQ_Z^\dagger(z) \nonumber\\
&\quad +\int\{\bar m^P_a(z)-\bar m^*_a(z)\}\,dQ_Z^\dagger(z) \nonumber\\
&\quad +\left[\int\bar m^*_a(z)\,dQ_Z^\dagger(z)-\int m_a(w)\,dQ_W^*(w)\right].
\label{eq:std_drift}
\end{align}
The first term is integrated conditional-compression drift under the observed law: by Theorem 2 it equals \(\int \operatorname{Cov}_P\{m_a(W),\pi_a(W)\mid Z=z\}\,/\,\E_P\{\pi_a(W)\mid Z=z\}\,dQ_Z^\dagger(z)\). The second term is conditional-law drift; it vanishes when \(P\) and \(Q^*\) induce the same conditional distribution of \(W\) given \(Z\). The third term is standardization drift; it vanishes when \(Q_Z^\dagger=Q^*_Z\), because then, by the tower property under \(Q^*\), \(\int\bar m^*_a(z)\,dQ_Z^\dagger(z)=\int m_a(w)\,dQ_W^*(w)\).
\end{corollary}

This decomposition clarifies the scope of Equation \eqref{eq:drift}. The covariance identity is a conditional drift identity under the observed law, and its integrated term accounts for the entire substitution drift exactly when the sum of the last two terms in Equation \eqref{eq:std_drift} is zero. A sufficient condition is that the observed and target laws induce the same conditional distribution of \(W\) given \(Z\) and that the standardizing distribution is the locked push-forward; these conditions make the two terms vanish individually. The terms may also cancel without vanishing individually. If their sum is nonzero, a standardized estimand can change through conditional-law or standardization mismatch even when the within-\(Z\) compression drift is zero.

\subsection{Design erasure as a special case}

The no-free-substitution theorem directly covers design erasure. Let \(W=(V,D)\), where \(D\) contains pre-action design or stratification information, and let \(Z=V\) be the compressed representation that drops \(D\). If
\[
A\indep Y(a)\mid (V,D),\qquad P(A=a\mid V,D)=\pi_a(V,D),
\]
then the conditional drift after dropping \(D\) is governed by
\[
\operatorname{Cov}\{m_a(V,D),\pi_a(V,D)\mid V\}.
\]
Thus exchangeability given \(V\) alone is guaranteed only under a causal compression condition, not by predictive performance. This is why generated covariates should usually augment rather than replace assignment-, availability-, eligibility-, observation-, or positivity-relevant variables.

\subsection{Mediator-as-state drift}

\begin{theorem}[Mediator-as-state drift]
Suppose \(M\) is a post-action mediator candidate with potential values \(M(a)\), and the locked total proximal effect at state \(v\) is
\[
\Delta_{\mathrm{tot}}(v)=\E\{Y(a,M(a))-Y(a',M(a'))\mid V=v\}.
\]
An analysis that conditions on observed \(M\) as if it were pre-action state generally does not identify \(\Delta_{\mathrm{tot}}(v)\). In the linear structural model
\[
M=\alpha_0+\alpha_A A+\alpha_V^\top V+\eta,\qquad
Y=\beta_0+\beta_A A+\beta_M M+\beta_V^\top V+\varepsilon,
\]
with \(\E(\eta\mid A,V)=0\) and \(\E(\varepsilon\mid A,M,V)=0\), the total effect for \(a-a'=1\) is \(\beta_A+\beta_M\alpha_A\), whereas the population regression coefficient of \(A\) in the regression of \(Y\) on \((1,A,M,V)\) is \(\beta_A\). The drift from the total effect is \(-\beta_M\alpha_A\).
\end{theorem}

\begin{proof}
Under the structural equations, setting \(A=a\) changes \(M\) by \(\alpha_A(a-a')\) and changes \(Y\) directly by \(\beta_A(a-a')\) and indirectly by \(\beta_M\alpha_A(a-a')\). For \(a-a'=1\), the total effect is \(\beta_A+\beta_M\alpha_A\). Moreover, \(\E(Y\mid A,M,V)=\beta_0+\beta_A A+\beta_M M+\beta_V^\top V\), so the population regression coefficient of \(A\) in the regression of \(Y\) on \((1,A,M,V)\) is \(\beta_A\). The two coincide only when the mediated path is null or cancels by special parameter restrictions.
\end{proof}

This theorem is not a claim that mediator analysis is invalid. It states that mediator adjustment targets a different estimand. A mediation estimand requires its own lock.

\subsection{Marker dependence is not marker intervention}

A state-dependent treatment contrast compares \(A=a\) with \(A=a'\) within levels of a state marker. It is not the effect of intervening on that marker.

\begin{theorem}[Marker-intervention non-equivalence]
There exist two structural causal models with the same observed distribution of \((U,A,Y)\) and the same state-dependent treatment effect \(\E\{Y(1)-Y(0)\mid U=u\}\), but different effects of intervening on the recorded marker \(U\).
\end{theorem}

\begin{proof}
Let \(A\sim\mathrm{Bernoulli}(1/2)\) be randomized and independent of all other exogenous variables. Consider binary \(U\) and a mean-zero error \(\varepsilon\). Model 1 has \(U\sim\mathrm{Bernoulli}(1/2)\) and
\[
Y=\beta A+\gamma U+\theta AU+\varepsilon.
\]
Here \(U\) is a direct cause of \(Y\). Model 2 has a latent \(L\sim\mathrm{Bernoulli}(1/2)\), \(U=L\), and
\[
Y=\beta A+\gamma L+\theta AL+\varepsilon.
\]
Here \(U\) is a marker of \(L\) and has no causal arrow into \(Y\). Because \(U=L\) observationally in Model 2, both models induce the same joint distribution of \((U,A,Y)\). In both models,
\[
\E(Y\mid A=1,U=u)-\E(Y\mid A=0,U=u)=\beta+\theta u.
\]
However, in Model 1, intervening to change \(U\) changes \(Y\) by \(\gamma+\theta A\). In Model 2, an intervention on the recorded marker \(U\) that leaves \(L\) unchanged does not change \(Y\). Thus the same state-dependent treatment effect is compatible with different marker-intervention effects.
\end{proof}

If a contextual cue appears to work better when a generated fatigue marker is high, that does not imply that manipulating the marker would improve outcomes; it is a next-trial intervention hypothesis.

\subsection{Post-action leakage}

\begin{proposition}[Post-action leakage changes the conditioning event]
Let \(Z\) be generated after treatment and have potential values \(Z(a)\). Assume consistency and unconditional randomization, \(A\indep\{Y(a),Y(a'),Z(a),Z(a')\}\), with positive allocation probabilities. At representation values where both arm-specific conditional means are defined, the observed conditional contrast
\[
\E(Y\mid A=a,Z=z)-\E(Y\mid A=a',Z=z)
\]
identifies
\[
\E\{Y(a)\mid Z(a)=z\}-\E\{Y(a')\mid Z(a')=z\},
\]
not the locked contrast \(\E\{Y(a)-Y(a')\mid V=v\}\) and not, in general, a principal-stratum contrast such as \(\E\{Y(a)-Y(a')\mid Z(a)=Z(a')=z\}\).
\end{proposition}

\begin{proof}
Consistency first gives \(\E(Y\mid A=a,Z=z)=\E\{Y(a)\mid A=a,Z(a)=z\}\). Unconditional randomization then removes \(A=a\) from the latter conditioning set. The comparison arm similarly gives \(\E\{Y(a')\mid Z(a')=z\}\). Unless \(Z(a)=Z(a')\) almost surely or special cross-world conditions are imposed, these are treatment-specific strata rather than a common pre-action state.
\end{proof}

For the state-dependent randomization used elsewhere in this paper, impose joint conditional randomization given the admissible pre-action state \(V\) instead. The corresponding identity is
\begin{align*}
&\E(Y\mid A=a,Z=z,V=v)-\E(Y\mid A=a',Z=z,V=v)\\
&\quad=\E\{Y(a)\mid Z(a)=z,V=v\}
-\E\{Y(a')\mid Z(a')=z,V=v\}.
\end{align*}
Dropping \(V\) is not justified by state-dependent randomization alone.

This proposition covers covariates extracted from post-treatment conversations, post-action wearable segments, post-session notes, or outcome-adjacent text. Such covariates may be important, but they require a different estimand. For text, conversation, and narrative representations, leakage may be semantic as well as temporal.

\subsection{Outcome-guided covariate selection and claim drift}

\begin{proposition}[Naive confirmatory claims after covariate selection are anti-conservative]
Let \(T_1,\ldots,T_p\) be independent standard normal statistics under \(p\) null effect-modification hypotheses, and let \(z_{0.975}=\Phi^{-1}(0.975)\approx 1.96\) denote the standard normal \(0.975\) quantile. Select \(J=\arg\max_j |T_j|\) and then test the selected covariate using the naive rule \(|T_J|>z_{0.975}\). The type I error probability is
\[
P(|T_J|>z_{0.975})=1-0.95^p,
\]
which exceeds 0.05 for every \(p>1\) and approaches one as \(p\) grows.
\end{proposition}

\begin{proof}
Under the global null, \(P(|T_j|\leq z_{0.975})=2\Phi(z_{0.975})-1=0.95\) exactly for each \(j\). Independence gives \(P(\max_j |T_j|\leq z_{0.975})=0.95^p\). The complement is the stated error probability.
\end{proof}

The calculation shows why the claim-status filter is part of the causal grammar. A generated covariate selected because it looked important is exploratory unless the analysis uses pre-specification, sample splitting, simultaneous inference, or a valid selective-inference procedure.

\subsection{Noisy state markers attenuate effect modification}

\begin{proposition}[Attenuation under classical state-measurement error]
Suppose \(A\) is randomized independently of \((U,e)\), \(\E(U)=\E(e)=0\), \(e\indep U\), and
\[
Y=\beta_0+\beta_AA+\lambda U+\theta AU+\varepsilon,\qquad U^{\obs}=U+e,
\]
with \(\E(\varepsilon\mid A,U,e)=0\). In the population linear projection of \(Y\) on \((1,A,U^{\obs},AU^{\obs})\), the coefficient of \(AU^{\obs}\) is
\[
\theta^{\obs}=\theta\frac{\operatorname{Var}(U)}{\operatorname{Var}(U)+\operatorname{Var}(e)}.
\]
\end{proposition}

\begin{proof}
Because \(A\) is binary and the projection includes \((1,A,U^{\obs},AU^{\obs})\), the population least-squares fit is saturated in \(A\): it is equivalent to fitting separate linear projections of \(Y\) on \((1,U^{\obs})\) within each arm. Within arm \(A=a\), the structural model gives \(Y=\beta_0+\beta_A a+(\lambda+\theta a)U+\varepsilon\), and since \(A\indep(U,e)\), \(e\indep U\), and \(\E(\varepsilon\mid A,U,e)=0\), the arm-specific projection slope on \(U^{\obs}=U+e\) is
\[
\frac{\operatorname{Cov}(Y,U^{\obs}\mid A=a)}{\operatorname{Var}(U^{\obs})}=(\lambda+\theta a)\,\kappa,
\qquad
\kappa=\frac{\operatorname{Var}(U)}{\operatorname{Var}(U)+\operatorname{Var}(e)}.
\]
The coefficient of \(AU^{\obs}\) in the joint projection equals the difference of the arm-specific slopes, \((\lambda+\theta)\kappa-\lambda\kappa=\theta\kappa\), which is the stated result.
\end{proof}

This proposition should not be overgeneralized. If measurement error is differential, if \(A\) is state-adaptive, or if the representation affects selection and missingness, representation error can induce bias beyond simple attenuation \cite{carroll2006}.

\section{Orthogonal estimation of the locked effect}

The type discipline determines which variables may enter the nuisance functions. Estimation can then use standard semiparametric machinery. The notation below distinguishes two targets that are often conflated: an empirical target evaluated at the observed target states, and a superpopulation target that also averages over the sampling distribution of those states.

Let \(V_{is}\) be the admissible pre-action state-design object. Define
\begin{align*}
\pi_a(v) &= P(A_{is}=a\mid V_{is}=v),\\
\rho_a(v) &= P(R_{is}=1\mid A_{is}=a,V_{is}=v),\\
\mu_a(v) &= \E(Y_{is}\mid A_{is}=a,R_{is}=1,V_{is}=v).
\end{align*}
Let \(h_{is}\geq0\) be target weights determined by a fixed, cluster-local rule. In the theorem below we use participant-normalized weights, \(\sum_{s=1}^{S_i}h_{is}=1\) almost surely. Let \(\mathcal F^-_{is}\) denote the information available immediately before the current allocation, and require
\[
\sigma(h_{is})\subseteq\sigma(V_{is})\subseteq\mathcal F^-_{is}.
\]
Thus \(V_{is}\) includes the pre-action information used to determine the weight, and the identification and nuisance-model conditions are imposed given that same \(V_{is}\). A prespecified rule cannot use a later session's state, allocation, or outcome to determine an earlier session's weight. For the stated theorem, the bounded schedule \(1\leq S_i\leq S_{\max}<\infty\) is fixed before participant \(i\)'s first allocation and is included in \(V_{is}\). The choice \(h_{is}=1/S_i\) is then allowed. Weighting by an outcome-dependent realized session count is not covered. Appendix A treats sample-normalized session weights separately as a ratio estimator.

For a single proximal outcome time, define the augmented score contribution
\begin{align}
\phi_{is}(a,a')
&=\frac{1(A_{is}=a)R_{is}}{\pi_a(V_{is})\rho_a(V_{is})}\{Y_{is}-\mu_a(V_{is})\}
-\frac{1(A_{is}=a')R_{is}}{\pi_{a'}(V_{is})\rho_{a'}(V_{is})}\{Y_{is}-\mu_{a'}(V_{is})\}\notag\\
&\quad +\mu_a(V_{is})-\mu_{a'}(V_{is}),
\label{eq:score}
\end{align}
and let \(\phi^{\mathrm{res}}_{is}(a,a')\) denote the same expression with the final regression term \(\mu_a(V_{is})-\mu_{a'}(V_{is})\) removed.

The cross-fitted estimator is
\begin{equation}
\widehat\Psi=\frac{1}{N}\sum_{i=1}^N\sum_{s=1}^{S_i}h_{is}\widehat\phi_{is}(a,a'),
\label{eq:estimator}
\end{equation}
where \(\widehat\phi\) uses nuisance estimates trained on folds that exclude participant \(i\).

For the random empirical target evaluated at the observed target states and weights,
\[
\Psi_N=\frac{1}{N}\sum_{i=1}^N\sum_{s=1}^{S_i}h_{is}\{\mu_a(V_{is})-\mu_{a'}(V_{is})\}.
\]
For a superpopulation target,
\[
\Psi=\E\left[\sum_{s=1}^{S_i}h_{is}\{\mu_a(V_{is})-\mu_{a'}(V_{is})\}\right].
\]
The corresponding cluster contributions are
\[
\varphi_i^{\emp}=\sum_{s=1}^{S_i}h_{is}\phi^{\mathrm{res}}_{is}(a,a'),\qquad
\varphi_i^{\supop}=\sum_{s=1}^{S_i}h_{is}\phi_{is}(a,a')-\Psi.
\]

\begin{theorem}[Orthogonal estimation for empirical and superpopulation targets]
Assume the identification conditions above, identically distributed independent clusters, and the predictable, participant-normalized weights and bounded predeclared schedules specified above. Assume finite second moments of the weighted outcome regressions, uniformly bounded conditional second moments of \(Y_{is}-\mu_b(V_{is})\) given \((A_{is}=b,R_{is}=1,V_{is})\), and positivity bounds \(\pi_b(v),\rho_b(v)\geq c>0\), for each arm \(b\in\{a,a'\}\). The representation and weighting rules are fixed independently of the evaluation data.

Use a fixed number of folds, each containing a positive limiting fraction of the clusters, and train each nuisance fit outside its evaluation fold. The fitted probabilities take values in \([c,1]\). Define the weighted norm by
\[
\|f\|_{2,h}^2=\E\left[\sum_{s=1}^{S_i}h_{is}f(V_{is})^2\right].
\]
For each fold and each arm \(b\), assume
\[
\|\widehat\mu_b-\mu_b\|_{2,h}
\bigl(\|\widehat\pi_b-\pi_b\|_{2,h}+\|\widehat\rho_b-\rho_b\|_{2,h}\bigr)
=o_p(N^{-1/2}),
\]
and individual consistency,
\[
\|\widehat\mu_b-\mu_b\|_{2,h}
+\|\widehat\pi_b-\pi_b\|_{2,h}
+\|\widehat\rho_b-\rho_b\|_{2,h}=o_p(1).
\]
Then, for the empirical target,
\[
\sqrt N(\widehat\Psi-\Psi_N)
=\frac{1}{\sqrt N}\sum_{i=1}^N\varphi_i^{\emp}+o_p(1)
\rightsquigarrow N\bigl(0,\E[(\varphi_i^{\emp})^2]\bigr).
\]
For the superpopulation target under the same cluster-local weighting rule,
\[
\sqrt N(\widehat\Psi-\Psi)
=\frac{1}{\sqrt N}\sum_{i=1}^N\varphi_i^{\supop}+o_p(1)
\rightsquigarrow N\bigl(0,\E[(\varphi_i^{\supop})^2]\bigr).
\]
The corresponding sample variance of the estimated cluster contributions consistently estimates the variance in each limit; the standard error of \(\widehat\Psi\) is its square root divided by \(\sqrt N\). Normal confidence intervals require a positive limiting variance. The empirical-target statement concerns marginal error about the random target \(\Psi_N\). It does not assert coverage conditional on the entire sequence of target states, because later states may contain earlier allocations and outcomes; such conditional coverage requires an additional conditional central limit theorem.
\end{theorem}

\begin{proof}
Predictability with respect to the same information used in the nuisance models gives \(\E[h_{is}\phi^{\mathrm{res}}_{is}\mid V_{is}]=0\). The weighted score is orthogonal to each nuisance perturbation. Cluster cross-fitting, the product-rate conditions, and individual consistency make the mean remainder and the centered score-estimation error negligible at the \(N^{-1/2}\) scale. The ordinary independent-cluster central limit theorem and cluster-score variance consistency give the two limits. Appendix B supplies the weighted moment calculation and remainder bound.
\end{proof}

In randomized trials, \(\pi_a(v)\) is often known by design. In observational sequential studies, it must be estimated or otherwise justified. The type discipline does not replace this requirement; it determines the pre-action information with respect to which the requirement is meaningful.

\section{Type uncertainty and audit discipline}

The classifier \(\calC\) is not assumed to be an oracle delivered by an AI system or by human auditors. It is a recorded causal judgement about a versioned covariate. The purpose of the audit is not to eliminate judgement but to expose it. A wrong type assignment can still produce a wrong locked analysis; the framework therefore treats type assignment as an auditable design decision, not as an automatically validated statistical output.

When a covariate could plausibly have more than one role, the primary analysis should use the most downstream role that can affect the estimand. A mixed pre-action and post-action transcript summary is treated as post-action unless the pre-action component is separable. A covariate derived from outcome-adjacent notes is treated as type \(Y\) unless an independent source window proves otherwise. A variable that determines availability, eligibility, randomization probability, or observation receives type \(D\) or Obs even if it is also predictive of the outcome. Outcome-trained representations are exploratory for effect modification unless the training data, selection rule, and inferential adjustment are external to or separated from the confirmatory analysis.

If auditors disagree about whether a covariate is a state marker, mediator, outcome proxy, or design variable, the primary analysis should either exclude the covariate from the confirmatory state or present separate locked estimands corresponding to the competing roles. An implementation may record a type-disagreement matrix, the conservative primary role, and the alternative locks used for sensitivity analysis. We do not model auditor disagreement as ordinary measurement error because a type disagreement is often a disagreement about the causal question itself.

\section{Stress-test simulation}

The simulation is included to illustrate the consequences of type violations, not to prove the theorems or establish clinical efficacy. The data-generating mechanism is a repeated-session setting with state-adaptive assignment, a design-relevant pre-action variable, a coarse state, a refined generated state, a post-action mediator, a post-action leakage covariate, and optional observed-state missingness. It is a reproducible artificial stress test reported using the ADEMP structure \cite{morris2019}.

For participant \(i\) and session \(s\), a design or availability variable \(D_{is}\) is generated before treatment. A latent state \(U^{\true}_{is}\) evolves from previous state, previous outcome, previous treatment, \(D_{is}\), and a participant covariate. A coarse pre-action state \(U^{(0)}\) is used by the assignment mechanism. A refined representation \(U^{(1)}\) measures the same underlying state with less noise. Treatment is assigned by
\[
P(A_{is}=1)=\expit\{-0.10+0.45U^{(0)}_{is}+0.80D_{is}+0.06(s-3.5)\},
\]
bounded away from zero and one. A post-action mediator is generated as
\[
M_{is}=0.50A_{is}+0.45U^{\true}_{is}+0.20C_i+0.25D_{is}+\zeta_{is},
\]
and the post-action leakage covariate is \(L_{is}=M_{is}+\kappa_{is}\). The outcome model is
\begin{align*}
Y_{is} &= b_i+0.45C_i+0.55U^{\true}_{is}+0.50D_{is}+0.20Y_{i,s-1}+0.05s \\
&\quad +A_{is}(0.20+0.45U^{\true}_{is})+0.50M_{is}+\varepsilon_{is}.
\end{align*}
The total effect at state \(U^{\true}\) is therefore \(0.45+0.45U^{\true}\). The target is the empirical standardized total effect across generated sessions. In the missing-outcome scenario, observation depends on observed pre-action and history variables, so the missingness mechanism is observed-state MAR by construction.

We compare seven specifications. The naive model uses only \(A\). The locked coarse-state model uses the design-relevant state variables and the coarse interaction with treatment. The predictive-refinement model replaces design-relevant information by the refined representation, illustrating design erasure. The admissible refinement keeps design information while adding \(U^{(1)}\) and allowing effect modification by \(U^{(1)}\). The mediator-adjusted model incorrectly includes post-action \(M\) in the primary total-effect regression. The leakage-adjusted model incorrectly includes post-action \(L\) as if it were state. The oracle model uses \(U^{\true}\) and is not available in practice. All simulation tables use 1000 Monte Carlo replications. For a coverage estimate \(\widehat p\) based on \(R\) replications, the Monte Carlo standard error is \(\{\widehat p(1-\widehat p)/R\}^{1/2}\); at nominal 95 percent coverage and \(R=1000\) this is approximately 0.0069. Small differences in coverage near the nominal level should therefore not be over-interpreted.

\begin{table}[htbp]
\centering
\caption{Controlled failure-mode simulation for the standardized total effect. Results are based on 1000 Monte Carlo replications with \(N=140\) participants and \(S=6\) sessions.}
\label{tab:sim-main}
\small
\setlength{\tabcolsep}{4.5pt}
\begin{tabular}{llrrrrr}
\toprule
Scenario & Method & True & Estimate & Bias & RMSE & Coverage \\
\midrule
Complete & Untyped naive contrast & 1.068 & 1.852 & 0.784 & 0.799 & 0.001 \\
Complete & Coarse locked state & 1.068 & 1.055 & -0.014 & 0.095 & 0.948 \\
Complete & Design erasure & 1.068 & 1.209 & 0.141 & 0.169 & 0.655 \\
Complete & Admissible refinement & 1.068 & 1.057 & -0.011 & 0.092 & 0.947 \\
Complete & Mediator-as-state & 1.068 & 0.667 & -0.402 & 0.412 & 0.009 \\
Complete & Post-action leakage & 1.068 & 0.842 & -0.227 & 0.247 & 0.358 \\
Complete & Oracle state & 1.068 & 1.064 & -0.004 & 0.089 & 0.946 \\
\addlinespace
Observed-state MAR & Untyped naive contrast & 1.069 & 1.859 & 0.791 & 0.809 & 0.004 \\
Observed-state MAR & Coarse locked state & 1.069 & 1.064 & -0.005 & 0.106 & 0.960 \\
Observed-state MAR & Design erasure & 1.069 & 1.214 & 0.145 & 0.177 & 0.727 \\
Observed-state MAR & Admissible refinement & 1.069 & 1.063 & -0.006 & 0.102 & 0.952 \\
Observed-state MAR & IPCW admissible refinement & 1.069 & 1.060 & -0.009 & 0.103 & 0.956 \\
Observed-state MAR & Mediator-as-state & 1.069 & 0.698 & -0.371 & 0.385 & 0.070 \\
Observed-state MAR & Post-action leakage & 1.069 & 0.864 & -0.205 & 0.231 & 0.552 \\
Observed-state MAR & Oracle state & 1.069 & 1.066 & -0.002 & 0.100 & 0.956 \\
\bottomrule
\end{tabular}
\end{table}

Table \ref{tab:sim-main} illustrates four points. First, state-adaptive assignment makes naive treatment contrasts unusable. Second, refined representation is not sufficient if it erases design-relevant information. Third, retaining the design state while adding the refined state performs similarly to the oracle benchmark in this simple setting. Fourth, mediator adjustment and post-action leakage target the wrong estimand and produce large negative bias for the total effect. These results are consistent with the drift theorems above. They should not be read as a comprehensive comparison of estimators.

\begin{table}[htbp]
\centering
\caption{Controlled no-free-substitution stress test. Results are based on 1000 Monte Carlo replications for each variant.}
\label{tab:sim-substitution}
\small
\begin{tabular}{llrrrrr}
\toprule
Variant & Method & True & Estimate & Bias & RMSE & Coverage \\
\midrule
Baseline & Design erasure & 1.068 & 1.210 & 0.142 & 0.172 & 0.668 \\
Baseline & Admissible refinement & 1.068 & 1.060 & -0.008 & 0.095 & 0.939 \\
Baseline & Oracle state & 1.068 & 1.066 & -0.002 & 0.093 & 0.932 \\
Positivity stress & Design erasure & 1.070 & 1.317 & 0.246 & 0.266 & 0.296 \\
Positivity stress & Admissible refinement & 1.070 & 1.057 & -0.013 & 0.103 & 0.940 \\
Positivity stress & Oracle state & 1.070 & 1.067 & -0.003 & 0.100 & 0.948 \\
Outcome compression & Design erasure & 1.062 & 1.116 & 0.054 & 0.108 & 0.894 \\
Outcome compression & Admissible refinement & 1.062 & 1.056 & -0.005 & 0.095 & 0.947 \\
Outcome compression & Oracle state & 1.062 & 1.063 & 0.001 & 0.092 & 0.949 \\
Nonclassical representation & Design erasure & 1.068 & 1.113 & 0.045 & 0.101 & 0.921 \\
Nonclassical representation & Admissible refinement & 1.068 & 1.067 & -0.001 & 0.091 & 0.952 \\
Nonclassical representation & Oracle state & 1.068 & 1.063 & -0.005 & 0.089 & 0.957 \\
\bottomrule
\end{tabular}
\end{table}

Table \ref{tab:sim-substitution} focuses on the no-free-substitution theorem. The positivity-stress variant strengthens the association between design variables and treatment assignment. The outcome-compression variant removes the direct effect of the design variable on the outcome mean, making design erasure less harmful. The nonclassical-representation variant makes the refined generated state biased by design status. The results follow the covariance logic in Equation \eqref{eq:drift}: design erasure is most damaging when omitted design information jointly predicts assignment and potential-outcome means. These stylized scenarios do not exhaust representation failure modes; they isolate structural violations so that the resulting estimand drift can be interpreted.

We also simulated claim-status drift under no true generated-covariate moderation. In each replication, \(K=25\) pre-action candidate covariates were generated with no true interaction. The same-data procedure selected the covariate with the largest absolute interaction statistic and reported the ordinary nominal 95 percent interval for that selected interaction. The split-sample procedure selected on one half of the data and estimated on the other half.

\begin{table}[htbp]
\centering
\caption{Controlled claim-status stress test under no true generated-covariate moderation. Results are based on 1000 replications.}
\label{tab:sim-selection}
\small
\begin{tabular}{lrrrr}
\toprule
Procedure & Mean estimate & RMSE & Coverage & Mean selected \(|t|\) \\
\midrule
Same data selection and inference & 0.004 & 0.198 & 0.273 & 2.270 \\
Split selection and independent inference & -0.003 & 0.121 & 0.942 & 0.776 \\
\bottomrule
\end{tabular}
\end{table}

The same-data procedure had severe undercoverage because the selected covariate was chosen using the outcome. The result illustrates why claim status is part of the causal type discipline rather than a reporting label added after analysis.

\section{Relationship to existing literature}

The framework relies on standard potential-outcome and longitudinal causal inference ideas \cite{robins1986,robins2000,hernan2020,pearl2009}. It is related to micro-randomized trials and causal excursion effects \cite{klasnja2015,boruvka2018,qian2021,qian2022}. That literature defines and estimates proximal or excursion effects in repeated decision settings. The present paper asks a prior question: when high-dimensional generated covariates are available before and after decision points, which of them may define the state for such effects, which change the estimand, and which can only support exploratory claims?

The closest mathematical antecedents of Theorem 2 and Corollary 1 are the deconfounding-score results of D'Amour and Franks \cite{damour2021} and Clivio et al. \cite{clivio2026}. Their covariance criteria already extend beyond balancing and prognostic scores. The present use of that mechanism is to audit replacement of sequential state-design information by a generated representation, and to distinguish compression drift from changes in the target conditional law or standardization measure. Our proposed contribution is the integration of causal types, claim status, a locked estimand, and audit rules for generated covariates; it is not discovery of the covariance mechanism or zero-covariance compression criterion.

The claim-status component is related to post-selection inference \cite{lee2016,kuchibhotla2022}. Its role in the causal type discipline is to prevent exploratory generated covariates from being reported as confirmatory effect modifiers without accounting for selection. The estimation section connects to semiparametric efficiency, double/debiased machine learning, and targeted learning \cite{tsiatis2006,chernozhukov2018,vdl2006}. The novelty is not the orthogonal score itself; it is the integration of orthogonal estimation with an estimand-preserving type discipline for generated covariates.

Generated covariates also connect to causal representation learning, text-as-data causal inference, and causal inference in natural language processing \cite{johansson2016,scholkopf2021,veitch2020,keith2020,egami2022,feder2022}. Much of that literature asks how to learn representations for causal adjustment, text-based measurement, or counterfactual prediction. Treatment leakage in text-based causal inference is especially close to the post-action leakage problem studied here: text or narrative representations may encode treatment, response, or outcome information while being used as pre-action confounders \cite{daoud2022}. When a design-relevant state is unmeasured rather than compressed away, proximal causal inference shows that suitably structured proxy variables can restore identification \cite{miao2018}; causal typing is complementary, because it records whether a generated variable is proposed as state, as a proxy for unmeasured state, or in a role that changes the estimand. The point here is narrower and operational: a representation cannot be used for a causal claim until its source window, design role, downstream status, and claim status are fixed.

A related line of work studies inference after machine-learning predictions, AI/ML-generated covariates, or generated regressors enter downstream analyses \cite{fong2021,battaglia2024,angelopoulos2023,escanciano2023,christensen2026}. Those methods address prediction error, first-stage uncertainty, measurement error, validation-sample calibration, or local robustness once a generated variable has been accepted for an analytic role. Duan and Pelger \cite{duanpelger} describe, in their public abstract, moment-specific bias correction using a small human-labeled calibration set; this addresses downstream inferential correction and complements the role-admissibility question here. The present framework asks the preceding question: which causal role the generated variable may play. Causal typing is not a substitute for generated-regressor inference. It decides whether the variable belongs in the causal analysis at all; generated-regressor methods then address uncertainty once that role has been accepted.

Clinical trial estimand guidance is also relevant because it emphasizes the alignment of trial objective, intercurrent events, analysis, and sensitivity analysis \cite{ich2019}. The present framework generalizes that discipline to generated representations: the estimand is locked before covariate use, and the audit records which outputs of the representation map are compatible with the lock. AI-specific trial reporting guidance, including CONSORT-AI and SPIRIT-AI, emphasizes transparent description of AI inputs, outputs, human-AI interaction, and error cases \cite{liu2020consort,rivera2020spirit}. Regulatory AI guidance similarly emphasizes clear context of use, risk-based credibility assessment, data governance, documentation, and life-cycle management \cite{fda2025ai,fdaema2026}. Causal typing addresses a complementary inferential question: once an AI-derived variable exists, what causal role may it play in the locked analysis?

\section{Implementation as an audit object}

For reproducibility, the type discipline should be recorded as an auditable object. The record states the estimand lock, action milestone, source windows, representation versions, causal types, claim status, admissibility decisions, and reasons for exclusion or downgrading. It supports analysis plans, data dictionaries, codebooks, pipelines, and software tools without fixing a specific software implementation.

Table \ref{tab:audit-short} gives a simulation-grounded audit record for variables in the stress-test data-generating mechanism in Section 9. It is not a clinical vignette and does not represent a real or fictional patient record.

\begin{longtable}{>{\raggedright\arraybackslash}p{0.16\linewidth}>{\raggedright\arraybackslash}p{0.31\linewidth}>{\raggedright\arraybackslash}p{0.08\linewidth}>{\raggedright\arraybackslash}p{0.33\linewidth}}
\caption{Simulation-grounded causal type audit for the locked standardized total proximal effect.}\label{tab:audit-short}\\
\toprule
Feature & DGP source & Type & Primary-lock use \\
\midrule
\endfirsthead
\toprule
Feature & DGP source & Type & Primary-lock use \\
\midrule
\endhead
\(D_{is}\) & Pre-action design or availability variable used by assignment. & \(D\) & Retain in \(V\); do not replace by a predictive representation. \\
\(U^{(0)}_{is}\) & Coarse pre-action state used by the assignment mechanism. & \(B/U\) & Retain as part of the locked state-design object. \\
\(U^{(1)}_{is}\) & Refined generated pre-action state marker. & \(U\) & May augment \(V\) only if \(D_{is}\) is retained; not a substitute for \(D_{is}\). \\
\(M_{is}\) & Post-action mediator generated after \(A_{is}\). & \(M\) & Exclude from primary total-effect adjustment; use under a mediation lock. \\
\(L_{is}\) & Post-action leakage covariate derived from \(M_{is}\). & \(M\)/Obs & Exclude from pre-action state; it conditions on treatment-specific post-action strata. \\
Selected interaction marker & Chosen using the same outcome data. & \(U\) & Exploratory only; not confirmatory without sample splitting, simultaneous inference, or valid selective inference. \\
\bottomrule
\end{longtable}

The entries are DGP-defined rather than patient-like examples. In the admissible-refinement specification, \(U^{(1)}_{is}\) augments the state only after \(D_{is}\) is retained. In the design-erasure specification, the same refinement is problematic because it substitutes for design-relevant information.

This table also marks the boundary with adjacent literatures. Generated-regressor methods begin after a generated variable has been accepted into an analysis; this audit decides whether that causal role is allowed. Text-as-data methods study how representations are measured or learned; the audit records whether the resulting representation is state, design information, mediator, outcome proxy, observation process, exploratory signal, or excluded material under the lock.

\section{Reverse audit for hypothesis generation}

The primary role of the type discipline is protective: it prevents a generated covariate from entering a confirmatory analysis in a role that would change the locked estimand. Protection, however, should not be confused with deletion. A covariate that is inadmissible for the current locked estimand may still be scientifically informative if the audit records why it failed. We call this secondary use a reverse audit. It creates a hypothesis queue for future work; it does not change the current estimand and does not certify a new causal effect.

Three cases are especially useful in sequential clinical and digital experiments. First, a post-action variable proposed as pre-action state may be a candidate mediator or early mechanism marker. It must be excluded from the primary total-effect adjustment set, but it can motivate a separate mediation lock or a mechanistic substudy. Second, a leakage feature extracted after assignment may indicate source-window contamination, outcome proxying, or a biologically interesting early response after the decision point. These possibilities have different interpretations and require temporal validation before use. Third, a generated representation that cannot safely replace design variables may reveal unrecorded decision rules, availability judgements, site practices, or tacit clinical information encoded in notes or workflows. Such signals are not proof that the representation is causally admissible; they are reasons to improve measurement, design documentation, or the next experiment.

\begin{longtable}{>{\raggedright\arraybackslash}p{0.21\linewidth}>{\raggedright\arraybackslash}p{0.26\linewidth}>{\raggedright\arraybackslash}p{0.21\linewidth}>{\raggedright\arraybackslash}p{0.21\linewidth}}
\caption{Type-violation logs as exploratory discovery signals. A signal is inadmissible for the current locked estimand but may guide a new lock, validation study, or next-trial design.}\label{tab:reverse}\\
\toprule
Audit finding & Why it is not confirmatory for the current lock & Exploratory signal & Required follow-up \\
\midrule
\endfirsthead
\toprule
Audit finding & Why it is not confirmatory for the current lock & Exploratory signal & Required follow-up \\
\midrule
\endhead
Post-action mediator proposed as pre-action state & Adjusting for it can remove part of the total effect or condition on a descendant of treatment & Candidate mechanism or pathway marker & New mediation estimand, temporal source-window validation, sensitivity analysis \\
Post-action leakage feature with strong prediction & The compared groups condition on treatment-specific post-action strata & Source contamination, outcome proxying, or early response marker & Pre-specified window, external validation, or new early-response estimand \\
Generated representation erases design information & Predictive sufficiency does not imply exchangeability or positivity sufficiency & Unrecorded availability rule, site practice, clinician judgement, or tacit workflow signal & Compare with recorded design variables; improve design documentation; test safe compression \\
Outcome-selected effect modifier & The same data selected and tested the feature & Candidate state dependence & Sample splitting, selective inference, external validation, or pre-specification in the next study \\
Generation-unstable generated marker & The representation is not version-stable under the lock & Measurement error, extraction drift, or ambiguous source semantics & Version locking, repeated extraction audit, validation data, measurement-error analysis \\
\bottomrule
\end{longtable}

A reverse audit makes weaker claims than causal discovery or trial design. A type violation can suggest a hidden mechanism, but it can also indicate an error in timing, documentation, coding, observation intensity, or feature extraction. Therefore the output of a reverse audit should be assigned claim status ``exploratory'' or ``next-trial'', never ``confirmatory'', unless a new estimand lock and a valid design or inferential adjustment are supplied. Turning a downgraded signal into a randomized treatment version requires operational treatment versions, feasibility constraints, and a separate design calculation. The value here is narrower: potentially useful signals are preserved without contaminating the current causal claim.

\section{Discussion}

Generated covariates can improve causal analyses, especially when raw longitudinal context is too complex for simple hand-coded variables. The danger is that richer prediction can obscure causal role. A covariate extracted by an AI system may be clinically meaningful and highly predictive, yet still be a mediator, outcome proxy, design variable, or exploratory selected covariate. Treating all generated covariates as ordinary covariates invites estimand drift.

The same problem also arises in human clinical interpretation. Clinicians routinely integrate symptoms, narrative, tone, expectation, family context, treatment history, and tacit impressions into coherent causal stories. This is a strength of clinical medicine, not an error by itself. The problem begins when such narratives are promoted to confirmatory causal claims without specifying whether the underlying variable is a treatment version, pre-action state, prior history, mediator, outcome proxy, observation process, intercurrent event, or design variable. The present framework therefore applies to both AI-generated covariates and clinician-generated causal narratives. Its purpose is not to suppress clinical judgement but to make the causal role of that judgement explicit before it enters analysis or treatment-effect interpretation.

This distinction is especially important for placebo, context, communication, and behavioral studies. A clinician may observe that patients who trust the clinician respond better, that reassurance after an explanation is associated with symptom improvement, or that a patient seems to improve after a particular visit context. These observations may be clinically valuable. They are not automatically intervention effects. Trust before the encounter may be a pre-action state or history; reassurance after the encounter may be a mediator; clinician enthusiasm may be part of the treatment version only if it is operationally assigned; a post-visit note may contain outcome information; and a site-specific communication pattern may be a design variable. Without typing, the analysis may turn an intuitive clinical explanation into a different estimand.

This conservatism should not be read as a rejection of high-resolution representations. A useful representation map can decompose coarse clinical or behavioral impressions into finer pre-action components, such as nutrient composition and timing rather than a single meal indicator, activity-pattern summaries rather than a vague fatigue label, or environment-specific signals rather than a generic context descriptor. When such components are pre-action, version-stable, and compatible with the state-design object in the lock, causal typing lets them remain as admissible state markers or effect modifiers. The filter thus removes covariates whose roles would change the estimand while preserving fine-grained generated covariates that make the candidate state representation more explicit.

The proposed causal type discipline is conservative, but its endpoint is not rejection. It does not discover causal structure, license interventions on state markers, solve unmeasured confounding, or make high-dimensional conditioning easy. Its purpose is to route human- or AI-generated causal candidates to the appropriate scientific destination: some enter the confirmatory estimand, some become prespecified secondary effect modifiers, some remain exploratory hypotheses, some motivate mediation estimands, some reveal design-audit signals, and some become candidates for future randomized intervention versions. Causal typing thus prevents premature confirmatory claims while preserving scientifically useful signals for the next appropriate inferential or experimental step.

The theory suggests a practical rule for analysis: generated representations should usually augment, rather than replace, design-relevant variables unless a safe-compression condition is justified. If a representation is intended to refine state, analysts should retain the variables that supported randomization, eligibility, positivity, and observation. Similarly, post-action generated covariates should be assigned to mediation, intercurrent-event, observation, or outcome roles rather than moved into pre-action adjustment sets.

Causal typing addresses structural role errors in generated covariates, not all errors introduced by representation learning or human interpretation. Even an admissible pre-action marker may be noisy, unstable across extraction runs, affected by model-version drift, differentially measured across treatment or observation regimes, or based on clinical judgement that varies across raters. Classical measurement error may attenuate effect-modification signals, whereas nonclassical or differential representation error can introduce bias. Generated-regressor methods, repeated extraction, validation-sample calibration, measurement-error correction, and orthogonal or doubly robust estimators remain useful after an admissible representation has been fixed; they do not, by themselves, make a drifting or mismeasured representation causally admissible.

The main limitation is that causal typing depends on human judgement. The framework does not make the classifier \(\calC\) correct by definition; it makes the judgement explicit, auditable, and tied to consequences for the estimand. Certification in a specific retrospective information environment, where many generated variables jointly reconstruct a pre-action information set, is a broader problem outside the local role theory developed here. A related limitation is temporal: real workflows may have several plausible action milestones, so source-window sensitivity or separate locks may be needed. A third limitation is that the orthogonal estimator assumes the relevant pre-action state-design object is sufficient for exchangeability and observation. Unmeasured state variables, interference, and missing-not-at-random mechanisms still require design changes or sensitivity analyses. Finally, the simulation is intentionally stylized. It isolates structural type violations in simple data-generating mechanisms where drift can be interpreted. Nonlinear or high-dimensional representation maps may make the same violations harder to diagnose, but they do not remove the need for causal typing.

In this framework, AI systems are representation-generating procedures, and clinical intuition is a source of candidate explanations. Neither output becomes a confirmatory causal claim until its variables have been assigned causal roles. Once the roles are specified, standard identification and semiparametric estimation tools can be applied to the correct target, and signals that are not yet confirmatory can still be carried forward through reverse audit. Without this step, the analysis may remain clinically plausible and statistically sophisticated while answering the wrong causal question.

\section{Conclusion}

This paper introduced causal typing for generated covariates in sequential experiments. The framework assigns generated covariates to causal roles, fixes an estimand lock, and restricts permissible claims before covariates enter confirmatory analysis. The estimand-preservation theorem shows that admissible causal-role assignments identify the locked standardized proximal effect under standard identification conditions. The drift theorems show why the restrictions matter: mediator adjustment, design erasure, post-action leakage, marker-intervention conflation, and outcome-guided covariate selection each alter either the estimand or the inferential claim. The safe-compression criterion from the deconfounding-score literature is used to audit when a learned representation may replace design-relevant information.

The framework was motivated by AI-derived covariates in repeated-session clinical studies, but its scope is broader. It applies to any setting in which data-derived representations are constructed from longitudinal event streams and used in causal inference. It also applies to human-generated clinical narratives when those narratives are encoded as variables, adjustment choices, subgroup explanations, or causal claims. Its goal is to make generated covariates analysis-ready without allowing them to silently rewrite the causal question. The audit also routes non-confirmatory signals into clearly labelled mechanism hypotheses, early-response signals, unrecorded-design signals, or next-trial candidates. That secondary use is a disciplined record for future work, not a shortcut from prediction or intuition to causal discovery.

\section*{Data and code availability}

The paper is a theoretical proposal with artificial simulations. No real patient data are used. The framework is software-agnostic. The supplementary reproducibility package includes simulation scripts, package requirements, random seeds, replicate-level results, summary tables with Monte Carlo standard errors, verification scripts, a clean execution log, SHA256 file hashes, license information, and citation metadata.

\appendix

\section{Participant and session targets}

Let \(\Delta_{is}=m_1(V_{is})-m_0(V_{is})\). A participant-weighted empirical target is
\[
\Psi_{\mathrm{part}}=\frac{1}{N}\sum_{i=1}^N\frac{1}{S_i}\sum_{s=1}^{S_i}\Delta_{is}.
\]
A session-weighted empirical target is
\[
\Psi_{\mathrm{sess}}=\frac{1}{\sum_iS_i}\sum_{i=1}^N\sum_{s=1}^{S_i}\Delta_{is}.
\]
They coincide only when \(S_i\) is constant or when participant-specific session counts are unrelated to the session-level effects in the relevant weighted sense. The estimand lock should specify which target is intended before analysis.

For \(\Psi_{\mathrm{part}}\), Theorem 5 applies to \(h_{is}=1/S_i\) when the schedule is known before the first allocation and included in \(V_{is}\). If the realized count depends on an earlier outcome, its reciprocal is generally not an admissible weight for that earlier session. Prespecifying the stopping rule does not make the final count predictable at every allocation; such stopping requires a separate observation or weighting analysis.

Even with predeclared schedules, the session-weighted estimator uses a common sample denominator. Write \(\overline S_N=N^{-1}\sum_iS_i\), \(T_i=\sum_s\phi_{is}(1,0)\), and \(U_i=\sum_s\phi^{\mathrm{res}}_{is}(1,0)\). Under the same bounded-schedule, moment, cross-fitting, and nuisance-rate conditions, the ratio estimator \(\widehat\Psi_{\mathrm{sess}}=\sum_{i,s}\widehat\phi_{is}(1,0)/\sum_iS_i\) satisfies
\[
\sqrt N(\widehat\Psi_{\mathrm{sess}}-\Psi_{\mathrm{sess}})
=\frac{1}{\overline S_N\sqrt N}\sum_i U_i+o_p(1)
\]
for the empirical target defined above. Define the superpopulation session target by
\[
\Psi_{\mathrm{sess}}^{\supop}=\frac{\E[\sum_s\Delta_{is}]}{\E(S_i)}.
\]
Its expansion is instead
\[
\sqrt N(\widehat\Psi_{\mathrm{sess}}-\Psi_{\mathrm{sess}}^{\supop})
=\frac{1}{\sqrt N}\sum_i
\frac{T_i-\Psi_{\mathrm{sess}}^{\supop}S_i}{\E(S_i)}+o_p(1).
\]
These formulas follow by expanding the ratio around its positive mean denominator. In particular, random session normalization contributes to the superpopulation influence function and cannot be treated as fixed cluster-local weighting.

\section{Proof details for orthogonal estimation}

For one treatment level \(a\), write
\[
\phi_{is,a}=\frac{1(A_{is}=a)R_{is}}{\pi_a(V_{is})\rho_a(V_{is})}
\{Y_{is}-\mu_a(V_{is})\}+\mu_a(V_{is}).
\]
Because \(h_{is}\) is measurable with respect to \(V_{is}\), at the true nuisance functions
\begin{align*}
&\E\left[h_{is}\frac{1(A_{is}=a)R_{is}}{\pi_a(V_{is})\rho_a(V_{is})}
\{Y_{is}-\mu_a(V_{is})\}\mid V_{is}\right]\\
&\quad=h_{is}\frac{P(A_{is}=a,R_{is}=1\mid V_{is})}
{\pi_a(V_{is})\rho_a(V_{is})}
\E[Y_{is}-\mu_a(V_{is})\mid A_{is}=a,R_{is}=1,V_{is}]=0.
\end{align*}
Summing over the predeclared schedule gives mean-zero weighted cluster residuals. The same argument holds conditional on each fold's training data because the weighting rule is fixed and clusters are independent.

For an outcome-regression perturbation \(\delta_\mu(V_{is})\), the derivative of the weighted expectation is
\[
\E\left[\sum_s h_{is}\left\{1-
\frac{1(A_{is}=a)R_{is}}{\pi_a(V_{is})\rho_a(V_{is})}\right\}
\delta_\mu(V_{is})\right]=0.
\]
Derivatives with respect to \(\pi_a\) and \(\rho_a\) vanish by the weighted residual identity above. In particular, an unweighted mean-zero calculation alone would not justify this step for arbitrary trajectory-dependent weights.

For each evaluation fold, conditioning on its training data and then on \(V_{is}\) gives the exact mean score error
\[
\E[h_{is}(\widehat\phi_{is,a}-\phi_{is,a})\mid V_{is},\text{training}]
=h_{is}(\widehat\mu_a-\mu_a)(V_{is})
\left\{1-\frac{\pi_a(V_{is})\rho_a(V_{is})}
{\widehat\pi_a(V_{is})\widehat\rho_a(V_{is})}\right\}.
\]
Since the probabilities are bounded and bounded away from zero, its summed expectation is bounded in absolute value by a constant times
\[
\|\widehat\mu_a-\mu_a\|_{2,h}
\bigl(\|\widehat\pi_a-\pi_a\|_{2,h}
+\|\widehat\rho_a-\rho_a\|_{2,h}\bigr).
\]
The comparison arm has the same bound. Individual nuisance consistency and the conditional residual-moment bound also imply \(L_2\) convergence of each estimated cluster score: use \((\sum_s h_{is}x_s)^2\leq\sum_s h_{is}x_s^2\). Conditional on the training sample, the centered evaluation-fold mean of this score error is therefore \(o_p(N^{-1/2})\). A fixed number of folds preserves this order. Subtracting either \(\Psi_N\) or \(\Psi\) gives the two displayed asymptotic expansions in Theorem 5. The independent-cluster central limit theorem and \(L_2\) convergence of the estimated cluster contributions establish the normal limits and variance consistency.

The predictability restriction is essential. With two sessions, let \(A_{i1}\) be Bernoulli \(1/2\), let \(U_i\) be an independent symmetric sign, and set \(Y_{i1}(0)=Y_{i1}(1)=U_i\) and \(Y_{i2}(0)=Y_{i2}(1)=0\), with complete observation. All treatment effects are zero. If later history is used to define
\[
h_{i1}=\tfrac12+\tfrac14(2A_{i1}-1)U_i,\qquad h_{i2}=1-h_{i1},
\]
the true first-session contrast score is \(2(2A_{i1}-1)U_i\), yet its weighted expectation is \(1/2\). These positive, normalized weights violate \(\sigma(h_{i1})\subseteq\sigma(V_{i1})\subseteq\mathcal F^-_{i1}\), and are excluded by the theorem.

\section{Simulation reproducibility}

The simulation tables reported in Tables 4-6 were regenerated from the supplementary reproducibility package using 1000 Monte Carlo replications, master seed 20260609, and a deterministic seed schedule. The reported summaries are computed from replicate-level CSV files: 15,000 rows for Table 4, 12,000 rows for Table 5, and 2,000 rows for Table 6.

For a coverage estimate based on \(R\) Monte Carlo replications, the Monte Carlo standard error is approximately \(\{\widehat p(1-\widehat p)/R\}^{1/2}\). At \(\widehat p=0.95\) and \(R=1000\), this is approximately 0.0069, so differences of roughly one percentage point around nominal coverage should be interpreted cautiously.

The reproducibility package includes the simulation script, an optional relative-path parallel runner for Tables 5 and 6, replicate-level results, summary tables, a clean execution log, a verification script, and SHA256 file hashes. The observation mechanism in the simulation script depends on observed pre-action and history variables. A latent-state observation mechanism is treated only as a latent-state sensitivity stress test, not as observed-data MAR.

\section{Checklist for a causal type audit}

A minimal audit should answer the following questions before outcome modeling begins.
\begin{enumerate}
\item What are the treatment versions \(a\) and \(a'\)?
\item What is the action time that separates pre-action from post-action information?
\item Does any source window overlap recommendation formation, assignment recording, patient communication, or intervention delivery? If so, is the covariate treated as mixed-window or analyzed under separate action-milestone locks?
\item What is the target distribution: participant-weighted, session-weighted, restricted, or transported?
\item Which generated covariates are treatment versions, core state, candidate state, history, design variables, mediators, outcomes, missingness indicators, intercurrent events, or excluded descriptors?
\item Which covariates were pre-specified and which were generated or selected after looking at outcomes?
\item Which design variables determine assignment, availability, eligibility, observation, or positivity?
\item Which post-action covariates are forbidden in the primary total-effect adjustment set?
\item What missingness and intercurrent-event strategies are part of the estimand lock?
\item What claims are confirmatory, secondary, exploratory, next-trial hypotheses, or non-causal descriptions?
\item Is the full representation map versioned so that another analyst can reproduce the covariates?
\item If a representation substitutes for design variables, what safe-compression argument is being made?
\item If type assignment is uncertain or auditors disagree, what conservative primary role was chosen, and which alternative locks are reported as sensitivity analyses?
\item Has the generated representation been checked for generation instability, model-version drift, measurement error, or differential extraction error?
\item If a covariate is excluded or downgraded because of a type violation, is the violation reason recorded?
\item Should the excluded or downgraded covariate be carried forward as a candidate mechanism, early-response marker, unrecorded-design signal, source-window problem, or next-trial intervention candidate?
\item What new estimand lock, temporal validation, external validation, sample splitting, selective inference, or randomized intervention version would be required before that signal could support a causal claim?
\end{enumerate}

\section{Structured audit-template fields}

A full audit template may be provided as supplementary material or in a public repository. The essential fields are: treatment versions, action time, outcome time, target distribution, intercurrent-event strategy, missingness strategy, claim status, feature name, source window, representation version, causal type, claim status for the feature, whether the feature is retained in the primary state-design object, whether it is admissible for the primary total-effect analysis, the violation that would occur if it were misused, the scientific follow-up, and the validation required before confirmatory use. This appendix records the fields without prescribing a particular file format.

\begingroup
\setstretch{1.0}
\endgroup

\end{document}